\documentclass[lettersize,journal]{IEEEtran}
\usepackage{amsmath,amsfonts}
\usepackage{array}
\usepackage[caption=false,font=normalsize,labelfont=sf,textfont=sf]{subfig}
\usepackage{textcomp}
\usepackage{stfloats}
\usepackage{url}
\usepackage{verbatim}
\usepackage{graphicx}
\usepackage{cite}

\usepackage{cite}
\usepackage{amsmath,amssymb,amsfonts}
\usepackage{algorithmic}
\usepackage{graphicx}
\usepackage{textcomp}
\usepackage{amsmath,amsfonts}
\usepackage[ruled,vlined]{algorithm2e}
\usepackage{array}
\usepackage[caption=false,font=normalsize,labelfont=sf,textfont=sf]{subfig}
\usepackage{textcomp}
\usepackage{stfloats}
\usepackage{url}
\usepackage{verbatim}
\usepackage{graphicx}
\usepackage{enumerate}
\usepackage{amsmath}
\usepackage{textcomp}
\usepackage{bm}

\def\mathbi#1{\textbf{\em #1}}
\newtheorem{definition}{Definition} 
\newtheorem{theorem}{Theorem}

\newtheorem{proof}{Proof}

\begin{document}

\title{Slightly Altruistic Nash Equilibrium for Multi-agent Pursuit-Evasion Games With Input Constraints}

\author{Dongting Li, dalton.dtli@gmail.com
}



\maketitle

\begin{abstract}
In this article, a class of input-constrained multi-agent pursuit-evasion (MPE) games with time-energy optimality is formulated. 
By introducing the cooperative behaviors for the players in the same team, the cooperative-noncooperative MPE games framework is provided.
Based on this, the optimal control policies for the index functions with the proposed altruism terms (including the team-interested and teammates-interested terms) constitute the slightly altruistic global Nash equilibrium.
The critic-actor reinforcement learning approach is adopted to solve the equilibrium.
Besides, the capture condition for each pursuer is provided, especially in the case of multiple players with poor mobility capture fewer players with strong mobility.
To guarantee the capture conditions being satisfied in real-time, a rolling horizon target selection method with a novel bi-layer topology is provided, in which the first-layer contains the communication information, the second-layer of game topology is used to guarantee the capture conditions.
The rolling horizon strategy leads to the time element is also embedded in the bi-layer topology.
The numerical simulations demonstrate the validity of the proposed theory.
\end{abstract}


\section{Introduction}

\textbf{\emph{Contributions:}} The main contributions of this article includes:
\begin{enumerate}
\def\labelenumi{\arabic{enumi})}
\item[i)]
  \textbf{Time-energy optimal MPE games with input constraint:}\\
  In this paper, a class of input-constrained multi-agent pursuit-evasion games is formulated, with the time-energy optimality is considered in the finite-time case.
  Due to the input saturations, the coupled HJI equations in the optimal control framework are solved by the critic-actor reinforcement learning approach;
  
\item[ii)]
  \textbf{Slightly altruistic Nash equilibrium:}\\
  Consider only the noncooperative MPE games are studied in existing papers, i.e., the lack of cooperative behaviors for the players in the same team.
  A cooperative-noncooperative framework is provided, by introducing the altruism terms including team-interested and teammates-interested terms, then a slightly altruistic global Nash equilibrium is proposed;
\item[iii)]
  \textbf{Bi-layer topology-based rolling horizon target selection:}\\
  To guarantee the capture, especially in the cases that multiple players with poor mobility capture fewer players with strong mobility, a rolling horizon target selection approach via a bi-layer topology is presented.
  The first layer contains the communication information, and the weights in the second layer, a.k.a, game topology are used to guarantee the capture via the proposed adjusting law.
  Besides, the time element is also considered in the game topology due to the rolling horizon strategy, i.e., with time converge to the terminal, each player will puts more efforts on pursuing the evaders easier to capture.
\end{enumerate}

\textbf{\emph{Organization:}} In this paper, the preliminaries including the graph theory and the dynamics are presented in Section II.
The proposed cooperative-noncooperative framework, as well as the index functions with slight altruism, and the slightly altruistic global Nash equilibrium are provided in Section III.
Section IV adopts the critic-actor reinforcement learning method to solve the equilibrium.
In Section V, the bi-layer topology, and the finite-time capture analysis are presented.
Section VI summarizes this paper.

\textbf{\emph{Notation:}} Throughout this paper, the superscript ${\cdot}_i^p$ denotes the $i$-th pursuer, ${\cdot}_j^e$ denotes the $j$-th evader, ${\cdot}_{ij}^{pe}$ denotes the cross relation between pursuer $i$ and evader $j$, ${\cdot}_{\mathcal{N}_p-i}^p$ denotes the neighboring pursuers of pursuer $i$, ${\cdot}_{\mathcal{N}_{pe}-i}^e$ denotes the neighboring evaders of pursuer $i$, ${\cdot}_{\mathcal{N}_{e}-j}^e$ denotes the neighboring evaders of evader $j$.  $\otimes$ denotes the Kronecker product. $I_{n}$ denotes the identify matrix of $n$-dimension, $0_{n}$ denotes the zero matrix of $n$-dimension. $E^n_{12} \in \mathbb{R}^{2n \times n}$ denotes $[I^{\mathrm{T}}_n, 0^{\mathrm{T}}_n]^{\mathrm{T}}$, $E^n_{22} \in \mathbb{R}^{2n \times n}$ denotes $[0^{\mathrm{T}}_n, I^{\mathrm{T}}_n]^{\mathrm{T}}$.

\section{Preliminaries}

\subsection{Graph Theory}
\label{Graph Theory}
Assume the multiple spacecraft PE games involve $N$ pursuers and $M$ evaders ($N \geq M$). The graph among the $N$ pursuers is defined by $\mathcal{G}_p = (V_p, E_p)$, where $V_p = \{ v_{p1}, \cdots, v_{pN}\}$ represents the nodes set, and $E_p \subseteq V_p \times V_p$ denotes the edges set.
Define $c^{p}_{ij}$ as the weight of  connection between the $i$-th and the $j$-th pursuer, note that $c^{p}_{ij} = 0$ if $(v_{pi}, v_{pj}) \not\in E_p$.
Define the $i$-th pursuer's in-degree as $d_i^p = \sum_{j=1}^{N} c^{p}_{ij}$, and the in-degree matrix as $D_p = \mathrm{diag} \{ c_i^p \}$.
Define the weighted adjacency matrix as $C_p = [ c^{p}_{ij} ]$, then the graph Laplacian matrix can be represented by $L_p = D_p - A_p$.

Similarly, the graph of the $M$ evaders is defined as $\mathcal{G}_e = (V_e, E_e)$, with the $V_e = \{ v_{e1}, \cdots, v_{eM}\}$ denotes the nodes set, $E_e \subseteq V_e \times V_e$ represents the edges set.
Define $c^{e}_{kl}$ as the weight of connection between the $k$-th and the $l$-th evader.
Define $d_k^e = \sum_{l=1}^{M} c^{e}_{kl}$ as the $k$-th evader's in-degree, $D_e = \mathrm{diag} \{ d_k^e \}$ as the in-degree matrix, $C_e = [ c^{p}_{kl} ]$ as the weighted adjacency matrix, $L_e = D_e - A_e$ as the graph Laplacian matrix.

Define $\mathcal{G}_{pe} = (V_{pe}, E_{pe})$ as the global graph of representing the interactions among all players, $c^{pe}_{ik}$ as the weight of the $i$-th pursuer relative to the $k$-th evader, $c^{ep}_{ki}$ as the weight of the the $k$-th evader relative to the $i$-th pursuer, note that $c^{pe}_{ik}$ is not always equal to $c^{ep}_{ki}$.
Define the in-degree of the $i$-th pursuer in $G_{pe}$ as $d_i^{pe} = \sum_{k=1}^{M} c^{pe}_{ik}$, the $k$-th evader in $G_{pe}$ as $d_k^{ep} = \sum_{i=1}^{N} c^{ep}_{ki}$.
Define $D_{pe} = \mathrm{diag} \{ d_i^{pe} \}$ and $D_{ep} = \mathrm{diag} \{ d_k^{ep} \}$ as the in-degree matrix, $C_{pe} = [ c^{pe}_{ik} ]$ and  $C_{ep} = [ c^{ep}_{ki} ]$ as the weighted adjacency matrix.

\subsection{Dynamical model}
Consider a class of input-constrained MPE games of $N$ pursuers versus $M$ evaders ($N\geq M$), the dynamics of pursuer $i$ described by a linear differential equation is given by
\begin{equation}{\label{eq1}}
\dot{\mathbi{x}}^{p}_{i} = A \mathbi{x}^{p}_{i} + B_i^{p} \mathbi{u}_{i},\  i \in \mathcal{N}^{p}, 
\end{equation}
where $\mathcal{N}^{p} \triangleq \{1, \cdots, N\}$ denotes the set of pursuers, $\mathbi{x}^{p}_{i} \in \mathbb{R}^{n}$ denotes the position state. $\mathbi{u}_i \in \Omega_{\mathbi{u}}$ is the control input vector with $\Omega_{\mathbi{u}}=\{ \mathbi{u} | \mathbi{u} \in \mathbb{R}^{m},\ |u_j| \leq u_{\max},\ j = 1,\cdots,m \}$ represents the input set satisfying the saturation constraint, where $u_j$ is the $j$-th element of $\mathbi{u}$ and $u_{\max}$ denotes the upper bound. Time-invariant matrices $A \in \mathbb{R}^{n \times n}$ and $B_i^{p} \in \mathbb{R}^{n \times m}$ are the state matrix and control input matrix, $[ A, B^p_i ]$ is assumed controllable.

Consider the following dynamics of the $j$-th evader, i.e.,
\begin{equation}{\label{eq2}}
\dot{\mathbi{x}}^{e}_{j} = A \mathbi{x}^{e}_{j} + B_{j}^{e} \mathbi{v}_{j}, \ j \in \mathcal{N}^{e}, 
\end{equation}
where $\mathcal{N}^{e} \triangleq \{1, \cdots, M\}$ denotes the set of evaders, the state vector $\mathbi{x}^{e}_{j} \in \mathbb{R}^{n}$, the input vector $\mathbi{v}_{j} \in \Omega_{\mathbi{v}}$, the saturated input set $\Omega_{\mathbi{v}}$, the saturation bound $v_{\max}$, and the input matrix $B_j^{e} \in \mathbb{R}^{n \times m}$ are defined similarly with \eqref{eq1}, respectively.

Based on the graph $G_p$ and $G_{pe}$, define the local error of pursuer $i$ relative to his neighbors as
\begin{equation}{\label{eq3}}
\bm{\delta}^{p}_{i} = \left[
\sum\limits_{k \in \mathcal{N}^p_{-i}} c^{p}_{ik} \left( \mathbi{x}^{p}_{i} - \mathbi{x}^{p}_{k} \right)^{\mathrm{T}},\  \sum\limits_{j \in \mathcal{N}^{e}} c^{pe}_{ij} \left( \mathbi{x}^{p}_{i} - \mathbi{x}^{e}_{j} \right)^{\mathrm{T}}
\right]^{\mathrm{T}} \in \mathbb{R}^{2n}
\end{equation}
where $ \mathcal{N}^p_{-i}$ is the subset of $ \mathcal{N}^p$ that excludes pursuer $i$.
Take the time derivative of $\bm{\delta}^{p}_i$ and combine with \eqref{eq1} and \eqref{eq2}, yields,
\begin{equation}{\label{eq4}}
 \begin{aligned}
\dot{\bm{\delta}}^{p}_{i} 
&=  \left[ \begin{array}{cc}
     A &0 \\
     0 & A
   \end{array} \right] {\bm{\delta}}^{p}_{i} + 
   \left[ \begin{array}{c}
     a_i^{p}B^{p}_i  \\
     a_i^{pe}B^{p}_i
   \end{array} \right] \mathbi{u}_{i} -
    \left[ \begin{array}{c}
    \sum\limits_{k \in \mathcal{N}_{-i}^p} c^{p}_{ik} B_k^p \mathbi{u}_{k}  \\
     \sum\limits_{j \in \mathcal{N}^{e}} c^{pe}_{ij} B_j^e \mathbi{v}_{j}
   \end{array} \right]
   \\
&\triangleq  \bar{A} {\bm{\delta}}^{p}_{i}  + \bar{B}^{p}_{i} \mathbi{u}_{i} -  \sum\limits_{k \in \mathcal{N}_{-i}^p} c^{p}_{ik} E^{n}_{12} B_k^p \mathbi{u}_{k} -  \sum\limits_{j \in \mathcal{N}^{e}} c^{pe}_{ij} E^{n}_{22} B_j^e \mathbi{v}_{j}.
\end{aligned}
\end{equation}

Similarly, define the local error of the $j$-th evader as
\begin{equation}{\label{eq5}}
\bm{\delta}^{e}_{j} = \left[
\sum\limits_{l \in \mathcal{N}^e_{-j}} c^{e}_{jl} \left( \mathbi{x}^{e}_{j} - \mathbi{x}^{e}_{l} \right)^{\mathrm{T}},\  \sum\limits_{i \in \mathcal{N}^{p}} c^{ep}_{ji} \left( \mathbi{x}^{e}_{j} - \mathbi{x}^{p}_{i} \right)^{\mathrm{T}}
\right]^{\mathrm{T}} \in \mathbb{R}^{2n}.
\end{equation}
Take the time derivative of $\bm{\delta}^{e}_j$ and combine with \eqref{eq1} and \eqref{eq2}, yields,
\begin{equation}{\label{eq6}}
 \begin{aligned}
\dot{\bm{\delta}}^{e}_{j} 
&=  \left[ \begin{array}{cc}
     A &0 \\
     0 & A
   \end{array} \right] {\bm{\delta}}^{e}_{j} + 
   \left[ \begin{array}{c}
     a_j^{e}B^{e}_j  \\
     a_j^{ep}B^{e}_j
   \end{array} \right] \mathbi{v}_{j} -
    \left[ \begin{array}{c}
    \sum\limits_{l \in \mathcal{N}^e_{-j}} c^{e}_{jl} B_l^e \mathbi{v}_{l}  \\
     \sum\limits_{i \in \mathcal{N}^{p}} c^{ep}_{ji} B_i^p \mathbi{u}_{i}
   \end{array} \right]
   \\
&\triangleq  \bar{A} {\bm{\delta}}^{e}_{j}  + \bar{B}^{e}_{j} \mathbi{v}_{j} -  \sum\limits_{l \in \mathcal{N}^e_{-j}} c^{e}_{jl} E^{n}_{12} B_l^e \mathbi{v}_{l} -  \sum\limits_{i \in \mathcal{N}^{p}} c^{ep}_{ji} E^{n}_{22} B_i^p \mathbi{u}_{i}.
\end{aligned}
\end{equation}
Compared with the literature, the local error relative with the cooperative teammates and the adversarial opponents are defined by augmentation rather than direct summation. 
With such definition, the relevant error with the teammates and the opponents of each player are decoupled in his error state, then the neighboring team interest could be represented by his individual interest, which is introduced in the next result.

Consider the center of the $i$-th pursuer's neighboring cooperative team, which is given by
\begin{equation}{\label{eq7}}
\bm{\zeta}_i^p = \frac{ \mathbi{x}^{p}_{i} + \sum\nolimits_{k \in \mathcal{N}_{-i}^p} c^{p}_{ik}\mathbi{x}^{p}_{k}  }{ 1+ a_i^p} \equiv \mathbi{x}^{p}_{i} - \frac{  {E^{n}_{12}}^{\mathrm{T}}  \bm{\delta}^{p}_{i} }{1+ a_i^p} \in \mathbb{R}^{n} .
\end{equation}
Similarly, the center of his neighboring adversarial team is denoted by
\begin{equation}{\label{eq8}}
\bm{\zeta}_i^{pe} = \frac{ \mathbi{x}^{p}_{i} + \sum\nolimits_{j \in \mathcal{N}^e} c^{pe}_{ij} \mathbi{x}^{e}_{j}  }{ 1+ a_i^{pe}} \equiv \mathbi{x}^{p}_{i} - \frac{  {E^{n}_{22}}^{\mathrm{T}}  \bm{\delta}^{p}_{i} }{1+ a_i^{pe}}\in \mathbb{R}^{n}.
\end{equation}
Based on \eqref{eq7} and \eqref{eq8}, the error between the center of the pursuer $i$'s cooperative team and the center of his adversarial team is derived by $\bm{\zeta}_i^p - \bm{\zeta}_i^{pe} \equiv \left(  \frac{{E^{n}_{22}}^{\mathrm{T}}}  {1+ a_i^{pe}} -  \frac{  {E^{n}_{12}}^{\mathrm{T}}   }{1+ a_i^{p}} \right)  \bm{\delta}^{p}_{i} $. 

Consider the center of the $j$-th evader's neighboring cooperative team which is described by
\begin{equation}{\label{eq9}}
\bm{\zeta}_j^e = \frac{ \mathbi{x}^{e}_{j} + \sum\nolimits_{l \in \mathcal{N}_{-j}^e} c^{e}_{jl}\mathbi{x}^{e}_{l}  }{ 1+ a_j^e} \equiv \mathbi{x}^{e}_{j} - \frac{  {E^{n}_{12}}^{\mathrm{T}}  \bm{\delta}^{e}_{j} }{1+ a_j^e}\in \mathbb{R}^{n}.
\end{equation}
Similarly, the center of his neighboring adversarial team is denoted by
\begin{equation}{\label{eq10}}
\bm{\zeta}_j^{ep} = \frac{ \mathbi{x}^{e}_{j} + \sum\nolimits_{i \in \mathcal{N}^p} c^{ep}_{ji} \mathbi{x}^{p}_{i}  }{ 1+ a_j^{ep}} \equiv \mathbi{x}^{e}_{j} - \frac{  {E^{n}_{22}}^{\mathrm{T}}  \bm{\delta}^{e}_{j} }{1+ a_j^{ep}}\in \mathbb{R}^{n}.
\end{equation}
Similarly, the error between the center of evader $j$'s cooperative team and the center of his adversarial team is given by $\bm{\zeta}_j^e - \bm{\zeta}_j^{ep} \equiv \left(  \frac{{E^{n}_{22}}^{\mathrm{T}}}  {1+ a_j^{ep}} -  \frac{  {E^{n}_{12}}^{\mathrm{T}}   }{1+ a_j^{e}} \right)  \bm{\delta}^{e}_{j} $.

\section{Slightly Altruistic Nash Equilibrium for Input-constrained MPE Games}
Consider a class of input-constrained MPE games of $N$ pursuers versus $M$ evaders, the objective of each evading player is to \emph{maximize} the distance relative with his neighboring pursuers, under the input saturation constraint.
Conversely, each pursuing evader needs to \emph{minimize} the distance relative with his neighboring evaders, under the input saturation constraint.
During the games, each player can choose to whether keeping the cohesion with his cooperative team, i.e., \emph{minimizing} the distance relative with his teammates.
Note that, based on the graphical game theory, such control objective leads to each player maximizing or minimizing his distance relative with the center of gravity of his neighbors rather than a specific player.
Thus, a strategy is provided in Section \ref{sec5} to dynamically adjust the weight of graph to achieve a target selection.

The following results provide novel constructions of the player's objective via index functions.
\subsection{Time-energy optimal index with slight altruism}
Since the studies of MPE games in the literature only consider the non-cooperative relations, i.e., each agent plays a non-cooperative game with his opponents and teammates simultaneously.
In fact, it is necessary to introduce cooperative behaviors for the players in the same team, which contributes to achieve the overall objective of MPE games. 

Based on this, a novel cooperative-noncooperative index with slight altruism is provided in this paper, i.e., not only the non-cooperative self-interested term, but also the cooperative team-interested term and the teammates-interested term are also involved in the index. 
In this paper, the team-interested and teammates-interested terms are called the altruism term.

To realize the time-energy optimality and slight altruism, define the index function for the $i$-th pursuer as follows
\begin{equation}{\label{eq11}}
\begin{aligned}
J_i^{p} = &
 \int_{0}^{t_f}
 \Big[ 
  \varXi_{\rm{sf}} ({\bm{\delta}}^{p}_{i}) + \varXi_{\rm{tm}} ({\bm{\delta}}^{p}_{i}) + \varXi_{\rm{te}} ( {\bm{\delta}}_{i}^{p}, {\bm{\delta}}_{\mathcal{G}_p-i}^{p}, {\bm{\delta}}^{e} ) \\
  &+ U(\mathbi{u}_i) + \sum\limits_{k \in \mathcal{N}_{-i}^p} c^{p}_{ik} U(\mathbi{u}_k) - \sum\limits_{j \in \mathcal{N}^e} c^{pe}_{ij} U(\mathbi{v}_j) + \rho_i
\Big] \mathrm{d}t\\
&+  \bar{\varXi}_{\rm{sf}} ({\bm{\delta}}^{p}_{i}(t_f)) + \bar{\varXi}_{\rm{tm}} ({\bm{\delta}}^{p}_{i} (t_f) ) + \bar{\varXi}_{\rm{te}} ( {\bm{\delta}}^{p}(t_f), {\bm{\delta}}^{e} (t_f)),
\end{aligned}
\end{equation}
where $\varXi_{\rm{sf}}$, $\varXi_{\rm{tm}}$, and $ \varXi_{\rm{te}}$ denote the self-interested term, the team-interested term, and the teammates-interested term, respectively, $\bar{\varXi}_{\rm{sf}}$, $\bar{\varXi}_{\rm{tm}}$, and $ \bar{\varXi}_{\rm{te}}$ represent their corresponding terminal terms. ${\bm{\delta}}^{p}$ and ${\bm{\delta}}^{e}$ denote the local error set of all pursuers and evaders, ${\bm{\delta}}_{\mathcal{G}_p-i}^{p}$ is the subset of ${\bm{\delta}}^p$ that excludes pursuer $i$. Motivated by \cite{ref1,ref2}, the energy term with input constraint is defined as $U(\mathbi{u}_i) = 2 \int_{0}^{\mathbi{u}_i} \left( u_{\max} \tanh^{-1} \left( {\bm{\nu}} / {u_{\max}} \right) \right)^{\mathrm{T}} R^p_i \ \mathrm{d} \bm{\nu}$,  where $\bm{\nu} \in \mathbb{R}^{m} $, $R^p_i$ is a diagonal positive-definite matrix, and the definition of $U(\mathbi{u}_k)$ and $U(\mathbi{v}_j)$ are similar. $\rho_i$ is a positive parameter to adjust the weight of time optimality.

{\textbf{Self-interested term:}} The self-interested term of pursuer $i$ is defined by
\begin{equation}{\label{eq12}}
\varXi_{\rm{sf}} ({\bm{\delta}}^{p}_{i}) = {\bm{\delta}^p_i}^{\mathrm{T}} Q^{p}_{i} {\bm{\delta}^{p}_i}
\end{equation}
where $Q_i^p = \mathrm{diag}\{\Lambda_i^p, \Lambda_i^{pe}\} \in \mathbb{R}^{2n \times 2n}$, $\Lambda_i^p \in \mathbb{R}^{n \times n} $ and $\Lambda_i^{pe} \in \mathbb{R}^{n \times n} $ are positive-definite weight matrices. 

{\textbf{Team-interested term:}} The team-interested term of pursuer $i$ is defined by
\begin{equation}{\label{eq13}}
\begin{aligned}
\varXi_{\rm{tm}} ({\bm{\delta}}^{p}_{i}) &=  { \left( \bm{\zeta}_i^p - \bm{\zeta}_i^{pe} \right)}^{\mathrm{T}} \Gamma_i^p { \left( \bm{\zeta}_i^p - \bm{\zeta}_i^{pe} \right)} \\
& =  { \left( \frac{  {E^{n}_{22}}\!^{\mathrm{T}}  \bm{\delta}^{p}_{i} }{1+ a_i^{pe}} -\frac{  {E^{n}_{12}}\!^{\mathrm{T}}  \bm{\delta}^{p}_{i} }{1+ a_i^{p}} \right)}^{\mathrm{T}} \!\! \Gamma_i^p { \left( \frac{  {E^{n}_{22}}\!^{\mathrm{T}}  \bm{\delta}^{p}_{i} }{1+ a_i^{pe}} -\frac{  {E^{n}_{12}}\!^{\mathrm{T}}\bm{\delta}^{p}_{i} }{1+ a_i^{p}} \right)}
\end{aligned}
\end{equation}
where $\Gamma_i^p$ represents the weight matrix, one can adjust the the weight of altruism through the selection of $\Gamma_i^p$. \eqref{eq13} represents the pursuer $i$ pays efforts to control the center of his neighboring cooperative team to pursuit the center of his neighboring adversarial team.

{\textbf{Teammates-interested term:}} The teammates-interested term of pursuer $i$ is described by
\begin{equation}{\label{eq14}}
\begin{aligned}
\varXi_{\rm{te}} ( {\bm{\delta}}^{p}, {\bm{\delta}}^{e} ) =\ &
    {{\bm{\delta}}^{p}}^{\mathrm{T}} (C_p \otimes I_{2n}) \breve{Q}_i^p {\bm{\delta}}^{p}\\
    & +
    \left[ \begin{array}{c}
     { \bm{\delta}_i^{p} } \\
     {{\bm{\delta}}^{e}} 
   \end{array} \right]^{\mathrm{T}}
    \left[ \begin{array}{cc}
     I_{2n} & \\
     & C_{pe} \otimes I_{2n}
   \end{array} \right]
   \breve{Q}_i^{pe}
  \left[ \begin{array}{c}
     { \bm{\delta}_i^{p} } \\
     {{\bm{\delta}}^{e}} 
   \end{array} \right]
\end{aligned}
\end{equation}
where the matrix $\breve{Q}_i^{p} \in \mathbb{R}^{2nN \times 2nN} $ has the following form
\begin{equation}{\label{eq15}}
\breve{Q}_i^{p} = \left[ \begin{array}{ccccc}
     \mu^p_i Q^{p}_{1} & &Q_{1i}^p  &  &  \\
                  & \ddots &\vdots   &  & \\
     Q_{i1}^p & \cdots & \mu_i^p Q_i^p & \cdots & Q_{iN}^p \\
      & &\vdots & \ddots & \\
      & &Q_{Ni}^p & & \mu_i^p Q^{p}_{N}
   \end{array} \right],
 \end{equation}
and $\breve{Q}_i^{pe} \in \mathbb{R}^{2n(M+1) \times 2n(M+1)}$ is defined by
\begin{equation}{\label{eq16}}
\breve{Q}_i^{pe} = \left[ \begin{array}{ccccc}
     \eta_i^p Q^{p}_{i} &  Q^{pe}_{i1} & Q^{pe}_{i2} & \cdots &  Q^{pe}_{iM} \\
      Q^{ep}_{1i}    &  \eta_i^p Q^{e}_{1}  &   &  & \\
     Q^{ep}_{2i} &  & \eta_i^p Q_2^e &  &  \\
      \vdots& & & \ddots & \\
     Q^{ep}_{Mi} & & & & \eta_i^p Q^{e}_{M}
   \end{array} \right],  
\end{equation}
where $\mu$ and $\eta$ are the parameters to adjust the weight of altruism, $Q^p_{ik}$ satisfying $\bm{\delta}_i^p Q^p_{ik} \bm{\delta}_k^p > 0$ denotes the cross weight matrix of the local errors between pursuer $i$ and pursuer $k$, $Q^{pe}_{ij}$($Q^{ep}_{ji}$) satisfying $\bm{\delta}_i^p Q^{pe}_{ij} \bm{\delta}_j^e > 0$($\bm{\delta}_j^e Q^{ep}_{ji} \bm{\delta}_i^p >0$) denote the cross weight matrices of the local errors between the pursuer $i$(evader $j$) and evader $j$(pursuer $i$).

By neglecting the terms independent with ${\bm{\delta}}_i^{p}$ and combining with \eqref{eq12}, \eqref{eq13}, \eqref{eq14}, then \eqref{eq11} is converted to
\begin{equation}{\label{eq17}}
\begin{aligned}
J_i^{p} \!=\! & 
 \int_{0}^{t_f} \!\!
 \left[ {\bm{\delta}^p_i}^{\mathrm{T}} \tilde{Q}^{p}_{i} {\bm{\delta}^{p}_i} + {\bm{\delta}^p_i}^{\mathrm{T}} \! \left( \sum\limits_{k \in \mathcal{N}_{-i}^p} \!\! c^{p}_{ik} Q_{ik}^p {\bm{\delta}^p_k}
  + \!  \sum\limits_{j \in \mathcal{N}^e} \!\! c^{pe}_{ij} Q_{ij}^{pe}  {\bm{\delta}^{e}_j} \!\!
 \right) \right.\\
 &\left.+ U(\mathbi{u}_i) + \sum\limits_{k \in \mathcal{N}_{-i}^p} c^{p}_{ik} U(\mathbi{u}_k) - \sum\limits_{j \in \mathcal{N}^e} c^{pe}_{ij} U(\mathbi{v}_j) + \rho_i
\right] \mathrm{d}t \\
&+{\bm{\delta}^p_i}^{\mathrm{T}}(t_f) \bar{\tilde{Q}}^{p}_{i} {\bm{\delta}^{p}_i}(t_f) \\
&+{\bm{\delta}^p_i}^{\mathrm{T}}(t_f)  \left( \sum\limits_{k \in \mathcal{N}_{-i}^p} c^{p}_{ik} \bar{Q}_{ik}^p {\bm{\delta}^p_k}(t_f)
  + \sum\limits_{j \in \mathcal{N}^e} c^{pe}_{ij} \bar{Q}_{ij}^{pe}  {\bm{\delta}^{e}_j}(t_f) 
 \right)
\end{aligned}
\end{equation}
where
\begin{equation*}
\tilde{Q}^{p}_{i} \!\!=\!\! \left[ \!\!\!\! \begin{array}{cc}
    (1\!+\!\mu_i^p \!+\! \eta_i^p) \Lambda_i^p+ \frac{\Gamma_i^p}{(1+a_i^p)^2} &\!\!\!\! \frac{-\Gamma_i^p}{(1+a_i^p)(1+a_i^{pe})} \\
     \frac{-\Gamma_i^p}{(1+a_i^p)(1+a_i^{pe})} &\!\!\!\! (1\!+\!\mu_i^p \!+\! \eta_i^p) \Lambda_i^{pe} + \frac{\Gamma_i^p}{(1+a_i^{pe})^2}
   \end{array} \!\!\!\! \right]
\end{equation*}
and $\bar{\tilde{Q}}^{p}_{i}$, $\bar{Q}_{ik}^p$, and $\bar{Q}_{ij}^{pe}$ denote the terminal matrix.
Note that $\tilde{Q}^{p}_{i}$ can be guaranteed positive-definite through the selection of $\Lambda_i^p$, $\Lambda_i^{pe}$ and $\Gamma_i^p$.

Similarly, define the index function for the $j$-th evader as follows
\begin{equation}{\label{eq18}}
\begin{aligned}
J_j^{e} = &
 \int_{0}^{t_f}
 \Big[ 
  -\varXi_{\rm{sf}} ({\bm{\delta}}^{e}_{j}) - \varXi_{\rm{tm}} ({\bm{\delta}}^{e}_{j}) - \varXi_{\rm{te}} ( {\bm{\delta}}_j^{e}, {\bm{\delta}}_{\mathcal{G}_e-j}^{e}, {\bm{\delta}}^{p} ) \\
  &+ U(\mathbi{v}_j) + \sum\limits_{l \in \mathcal{N}_{-j}^e} c^{e}_{jl} U(\mathbi{v}_l) - \sum\limits_{i \in \mathcal{N}^p} c^{ep}_{ji} U(\mathbi{u}_i) + \rho_j
\Big] \mathrm{d}t\\
&-  \bar{\varXi}_{\rm{sf}} ({\bm{\delta}}^{p}_{i}(t_f)) - \bar{\varXi}_{\rm{tm}} ({\bm{\delta}}^{p}_{i} (t_f) ) - \bar{\varXi}_{\rm{te}} ( {\bm{\delta}}^{p}(t_f), {\bm{\delta}}^{e} (t_f)),
\end{aligned}
\end{equation}
where $\varXi_{\rm{sf}} ({\bm{\delta}}^{e}_{j})$, $ \varXi_{\rm{tm}} ({\bm{\delta}}^{e}_{j}) $, and $\varXi_{\rm{te}} ( {\bm{\delta}}_j^{e}, {\bm{\delta}}_{\mathcal{G}_e-j}^{e}, {\bm{\delta}}^{p} )$ are defined similarly with \eqref{eq13}, \eqref{eq14}, and \eqref{eq15}, respectively. Similar with \eqref{eq17}, \eqref{eq18} is converted to
\begin{equation}{\label{eq19}}
\begin{aligned}
J_j^{e} \!=\! & 
 \int_{0}^{t_f} \!\!
 \left[ -{\bm{\delta}^e_j}^{\mathrm{T}} \tilde{Q}^{e}_{j} {\bm{\delta}^{e}_j} \!-\! {\bm{\delta}^e_j}^{\mathrm{T}} \! \left( \sum\limits_{l \in \mathcal{N}_{-j}^e} \!\! c^{e}_{jl} Q_{jl}^e {\bm{\delta}^e_l}
  + \!\!  \sum\limits_{i \in \mathcal{N}^p} \!\! c^{ep}_{ji} Q_{ji}^{ep}  {\bm{\delta}^{p}_i} \!
 \right) \right.\\
 &\left.+ U(\mathbi{v}_j) + \sum\limits_{l \in \mathcal{N}_{-j}^e} c^{e}_{jl} U(\mathbi{v}_l) - \sum\limits_{i \in \mathcal{N}^p} c^{ep}_{ji} U(\mathbi{u}_i) + \rho_j
\right] \mathrm{d}t \\
&-{\bm{\delta}^e_j}^{\mathrm{T}}(t_f) \bar{\tilde{Q}}^{e}_{j} {\bm{\delta}^{e}_j}(t_f) \\
&-{\bm{\delta}^e_j}^{\mathrm{T}}(t_f)  \left( \sum\limits_{l \in \mathcal{N}_{-j}^e} \!\! c^{e}_{jl} Q_{jl}^e {\bm{\delta}^e_l}(t_f)
  +  \sum\limits_{i \in \mathcal{N}^p} \!\! c^{ep}_{ji} Q_{ji}^{ep}  {\bm{\delta}^{p}_i}(t_f) 
 \right),
\end{aligned}
\end{equation}
where
\begin{equation*}
\tilde{Q}^{e}_{j} \!\!=\!\! \left[ \!\!\!\! \begin{array}{cc}
    (1\!+\!\mu_j^e \!+\! \eta_j^e) \Lambda_j^e+ \frac{\Gamma_j^e}{(1+a_j^e)^2} &\!\!\!\! \frac{-\Gamma_j^e}{(1+a_j^e)(1+a_j^{ep})} \\
     \frac{-\Gamma_j^e}{(1+a_j^e)(1+a_j^{ep})} &\!\!\!\! (1\!+\!\mu_j^e \!+\! \eta_j^e) \Lambda_j^{ep} + \frac{\Gamma_j^e}{(1+a_j^{ep})^2}
   \end{array} \!\!\!\! \right]
\end{equation*}
with all the parameters and the weight matrices other than $\Lambda_j^e$ are defined similarly with \eqref{eq17}. Note that the matrix $\Lambda_j^e$ is defined negative-definite which means the $j$-th evader wishes to maximize his position relative to the neighboring pursuers, while keeping close with his teammates.

\subsection{Input-constrained MPE games formulation}
Based on the proposed index functions with slight altruism, the next result gives the definition of slightly altruistic global Nash equilibrium.
\begin{definition}[Slightly altruistic global Nash equilibrium]
\label{definition1}
Consider the MPE games of $N$ pursuers versus $M$ evaders, for $i \in \mathcal{N}^p$ and $j \in \mathcal{N}^e$, an $N+M$ tuple control policies $\{\mathbi{u}^*_1,\cdots, \mathbi{u}^*_N, \mathbi{v}^*_1, \cdots, \mathbi{v}^*_M  \}$ is called slightly altruistic global Nash equilibrium, with satisfying the saturation constraints $\mathbi{u}_i \in \Omega_{\mathbi{u}}$ and $\mathbi{v}_j \in \Omega_{\mathbi{v}}$, such that
\begin{equation}{\label{eq20}}
\begin{aligned}
 J_{i}^{p} ({ \mathbi{u}^{*}_i }, {\left. {\mathbi{u}}^{*}_{\mathcal{G}_p-i} \right.}, {\left. {\mathbi{v}}^{*}_{\mathcal{G}_{pe}-i} \right.}) &\leq J_{i}^{p} (  { \mathbi{u}_i }, {\left. {\mathbi{u}}^{*}_{\mathcal{G}_p-i} \right.}, {\left. {\mathbi{v}}^{*}_{\mathcal{G}_{pe}-i} \right.})\\
J_{j}^{e} ({ \mathbi{v}^{*}_j }, {\left. {\mathbi{v}}^{*}_{\mathcal{G}_e-j} \right.}, {\left. {\mathbi{u}}^{*}_{\mathcal{G}_{pe}-j} \right.}) &\leq J_{i}^{p} (  { \mathbi{v}_i }, {\left. {\mathbi{v}}^{*}_{\mathcal{G}_e-j} \right.}, {\left. {\mathbi{u}}^{*}_{\mathcal{G}_{pe}-j} \right.})
\end{aligned}
\end{equation}
where ${\left. {\mathbi{u}}^{*}_{\mathcal{G}_p-i} \right.}$ denotes the optimal control set for players in $\mathcal{N}^p$ escaping pursuer $i$, ${\left. {\mathbi{v}}^{*}_{\mathcal{G}_{pe}-i} \right.}$ is the optimal control set of pursuer $i$'s neighboring players  in $\mathcal{N}^e$, ${\left. {\mathbi{u}}^{*}_{\mathcal{G}_e-j} \right.}$ and ${\left. {\mathbi{u}}^{*}_{\mathcal{G}_{pe}-j} \right.}$ have similar definition.
\end{definition}

Note that the existing global Nash equilibrium is corresponding to the problem of multi-agent non-cooperative games, in which each player only considers his self interest.
In this paper, the altruism terms are introduced in the index functions, then a cooperative-noncooperative equilibrium could be studied based on the existing Nash framework.

\subsection{Solution for input-constrained MPE games}
Based on Definition \ref{definition1}, the next results give the solution for the input-constrained MPE games.
Consider the control policies of all players are coupled, which is reflect by coupled Hamilton–Jacobi–Isaacs (HJI) equations.
To obtain the form of HJI equations, each player's optimal control policy is derived firstly.

Consider the index functions \eqref{eq17} and \eqref{eq19} for pursuer $i$ and evader $j$, then the Hamiltonian functions are defined by
\begin{equation}{\label{eq21}}
\begin{aligned}
H_i^p =\ &   {\bm{\delta}^p_i}^{\mathrm{T}} \tilde{Q}^{p}_{i} {\bm{\delta}^{p}_i} + {\bm{\delta}^p_i}^{\mathrm{T}}  \left( \sum\limits_{k \in \mathcal{N}_{-i}^p} c^{p}_{ik} Q_{ik}^p {\bm{\delta}^p_k}
  + \sum\limits_{j \in \mathcal{N}^e} c^{pe}_{ij} Q_{ij}^{pe}  {\bm{\delta}^{e}_j} 
 \right) \\
 &+ U(\mathbi{u}_i) + \sum\limits_{k \in \mathcal{N}_{-i}^p} c^{p}_{ik} U(\mathbi{u}_k) - \sum\limits_{j \in \mathcal{N}^e} c^{pe}_{ij} U(\mathbi{v}_j) + \rho_i \\
                & +{\nabla{J_i^{p}}}^{\mathrm{T}}   \left[ \bar{A} {\bm{\delta}}^{p}_{i}  + \bar{B}^{p}_{i} \mathbi{u}_{i} -  \sum\limits_{k \in \mathcal{N}_{-i}^p} c^{p}_{ik} E^{n}_{12} \bar{B}_k^p \mathbi{u}_{k}\right. \\
                &\left. -  \sum\limits_{j \in \mathcal{N}^{e}} c^{pe}_{ij} E^{n}_{22} \bar{B}_j^e \mathbi{v}_{j} \right],
\end{aligned}
\end{equation}
and
\begin{equation}{\label{eq22}}
\begin{aligned}
H_j^e = &    -{\bm{\delta}^e_j}^{\mathrm{T}} \tilde{Q}^{e}_{j} {\bm{\delta}^{e}_j} - {\bm{\delta}^e_j}^{\mathrm{T}} \! \left( \sum\limits_{l \in \mathcal{N}_{-j}^e} \!\! c^{e}_{jl} Q_{jl}^e {\bm{\delta}^e_l}
  +  \sum\limits_{i \in \mathcal{N}^p}  c^{ep}_{ji} Q_{ji}^{ep}  {\bm{\delta}^{p}_i} 
 \right) \\
 &+ U(\mathbi{v}_j) + \sum\limits_{l \in \mathcal{N}_{-j}^e} c^{e}_{jl} U(\mathbi{v}_l) - \sum\limits_{i \in \mathcal{N}^p} c^{ep}_{ji} U(\mathbi{u}_i) + \rho_j
\\
                & +{\nabla{J_j^{e}}}^{\mathrm{T}}   \left[  \bar{A} {\bm{\delta}}^{e}_{j}  + \bar{B}^{e}_{j} \mathbi{v}_{j} -  \sum\limits_{l \in \mathcal{N}^e_{-j}} c^{e}_{jl} E^{n}_{12} \bar{B}_l^e \mathbi{v}_{l} \right. \\
                &\left. -  \sum\limits_{i \in \mathcal{N}^{p}} c^{ep}_{ji} E^{n}_{22} \bar{B}_i^p \mathbi{u}_{i} \right],           
\end{aligned}
\end{equation}
where ${\nabla{J_i^{p}}} \triangleq {\partial J_i^p}/{\partial \bm{\delta}_i^p}$ and ${\nabla{J_j^{e}}} \triangleq {\partial J_j^e}/{\partial \bm{\delta}_j^e}$.

By defining the corresponding optimal value functions as $V_i^{p} = \min\nolimits_{\mathbi{u}_i} J_i^p$ and $V_j^{e}= \min\nolimits_{\mathbi{v}_j} J_j^e $, 
the Hamilton-Jacobi (HJ) equations are derived as
\begin{equation}{\label{eq23}}
\begin{aligned}
- \frac{\partial V_i^p}{\partial \bm{\delta}_i^p} = \min\limits_{\mathbi{u}_i} H_i^p \left( V_i^p, \bm{\delta}_i^p, \bm{\delta}_{\mathcal{G}_p-i}^p, \bm{\delta}_{\mathcal{G}_{pe}-i}^e,  { \mathbi{u}_i },  {\mathbi{u}}^{*}_{\mathcal{G}_p-i}, {\mathbi{v}}^{*}_{\mathcal{G}_{pe}-i}  \right),\\
 i \in \mathcal{N}^p,
\end{aligned}
\end{equation}
and
\begin{equation}{\label{eq24}}
\begin{aligned}
- \frac{\partial V_j^e}{\partial \bm{\delta}_j^e} = \min\limits_{\mathbi{u}_i} H_j^e \left( V_j^e, \bm{\delta}_j^e, \bm{\delta}_{\mathcal{G}_e-j}^e, \bm{\delta}_{\mathcal{G}_{pe}-j}^p, { \mathbi{u}_i },  {\mathbi{u}}^{*}_{\mathcal{G}_p-i}, {\mathbi{v}}^{*}_{\mathcal{G}_{pe}-j}  \right), \\
 j \in \mathcal{N}^e.
\end{aligned}
\end{equation}

Then the optimal control policies are derived by
\begin{equation}{\label{eq25}}
\frac{	 \partial{{H}_i^p}  }{  \partial{ {\bm{u}_i^{*}}	} } = 0 \Rightarrow {\bm{u}_i^{*}} = -u_{\max}\tanh \left(\frac{1}{2u_{\max}} {R_i^p}^{-1} {\left.\bar{B}_i^p\right.}^{\mathrm{T}} {\nabla{V_i^{p}}} \right),
\end{equation}
and
\begin{equation}{\label{eq26}}
\frac{	 \partial{{H}_j^e}  }{  \partial{ {\bm{v}_j^{*}}	} } = 0 \Rightarrow {\bm{v}_j^{*}} = -v_{\max}\tanh \left(\frac{1}{2v_{\max}} {R_j^e}^{-1} {\left.\bar{B}_j^e\right.}^{\mathrm{T}} {\nabla{V_j^{e}}} \right).
\end{equation}

Based on \eqref{eq25} and \eqref{eq26}, it holds that
\begin{equation}{\label{eq27}}
\begin{aligned}
U({\mathbi{u}}^*_i) = &\ u_{\max} {V_i^p}^{\mathrm{T}} \bar{B}_i^p \tanh( {{D}_i^p}^{*})\\
                                 & + u^2_{\max} \bar{R}_i^p \ln \left( \underbar{\mathbi{1}} - \tanh^2({{D}_i^p}^{*})\right),
\end{aligned}
\end{equation}
and
\begin{equation}{\label{eq28}}
\begin{aligned}
U({\mathbi{v}}^*_j) = &\ v_{\max} {V_j^e}^{\mathrm{T}} \bar{B}_j^e \tanh( {{D}_j^e}^{*})\\
                                 & + v^2_{\max} \bar{R}_j^e \ln \left( \underbar{\mathbi{1}} - \tanh^2({{D}_j^e}^{*})\right),
\end{aligned}
\end{equation}
where $D_i^p = (1/2u_{\max}) {R_i^p}^{-1} {\left.\bar{B}_i^p \right.}^{\mathrm{T}} {\nabla{J_i^{p}}} $, $D_j^e = (1/2v_{\max}) {R_j^e}^{-1} {\left.\bar{B}_j^e\right.}^{\mathrm{T}} {\nabla{J_j^{e}}} $, $\bar{R}_i^p = \left[ r^p_{i1}, \cdots, r^p_{in} \right] \in \mathbb{R}^{1 \times n}$ with each $r^p_{ik}$ ($k \in[1,n]$) is the diagonal element of $R_i^p$, $\bar{R}_j^e$ has the same definition, and $\underbar{\mathbi{1}}$ is a column vector with an appropriate dimension and all elements are scalar $1$.

{\textbf{Coupled HJI equations:}} Substituting \eqref{eq25}, \eqref{eq27} into \eqref{eq23}, then the coupled HJI equation for pursuer $i$ is derived by
\begin{equation}{\label{eq29}}
\begin{aligned}
&{\nabla{V_i^{p}}}^{\mathrm{T}}   \bar{A} {\bm{\delta}}^{p}_{i} + {\bm{\delta}^p_i}^{\mathrm{T}} \bar{Q}^{p}_{i} {\bm{\delta}^{p}_i} + {\bm{\delta}^p_i}^{\mathrm{T}}  \left(
 \sum\limits_{k \in \mathcal{N}_{-i}^p} c^{p}_{ik} Q_{ik}^p {\bm{\delta}^p_k} \right.\\
&  \left. + \sum\limits_{j \in \mathcal{N}^e} c^{pe}_{ij} Q_{ij}^{pe}  {\bm{\delta}^{e}_j} 
 \right) + u^2_{\max} \bar{R}_i^p \ln \left( \underbar{\mathbi{1}} - \tanh^2({{D}_i^p}^{*})\right)\\
 &+ \sum\limits_{k \in \mathcal{N}_{-i}^p} c^{p}_{ik} u^2_{\max} \bar{R}_k^p \ln \left( \underbar{\mathbi{1}} - \tanh^2({{D}_k^p}^{*})\right)\\
 &-  \sum\limits_{j \in \mathcal{N}^{e}} c^{pe}_{ij} v^2_{\max} \bar{R}_j^e \ln \left( \underbar{\mathbi{1}} - \tanh^2({{D}_j^e}^{*})\right) + \rho_i = 0.
\end{aligned}
\end{equation}
Similarly, the HJI equation for evader $j$ is derived by
\begin{equation}{\label{eq30}}
\begin{aligned}
&{\nabla{V_j^{e}}}^{\mathrm{T}}   \bar{A} {\bm{\delta}}^{e}_{j} - {\bm{\delta}^e_j}^{\mathrm{T}} \bar{Q}^{e}_{j} {\bm{\delta}^{e}_j} - {\bm{\delta}^e_j}^{\mathrm{T}}  \left(
 \sum\limits_{l \in \mathcal{N}_{-j}^e} c^{e}_{jl} Q_{jl}^e {\bm{\delta}^e_l} \right.\\
&  \left. + \sum\limits_{i \in \mathcal{N}^p} c^{ep}_{ji} Q_{ji}^{ep}  {\bm{\delta}^{p}_i} 
 \right) + v^2_{\max} \bar{R}_j^e \ln \left( \underbar{\mathbi{1}} - \tanh^2({{D}_j^e}^{*})\right)\\
 &+ \sum\limits_{l \in \mathcal{N}_{-j}^e} c^{e}_{jl} v^2_{\max} \bar{R}_l^e \ln \left( \underbar{\mathbi{1}} - \tanh^2({{D}_l^e}^{*})\right)\\
 &-  \sum\limits_{i \in \mathcal{N}^{p}} c^{ep}_{ji} u^2_{\max} \bar{R}_j^e \ln \left( \underbar{\mathbi{1}} - \tanh^2({{D}_i^p}^{*})\right) + \rho_j = 0.
\end{aligned}
\end{equation}

\begin{theorem}[Time-energy optimal solution for input-constrained MPE games]
Consider the input-constrained MPE games of $N$ pursuers (with dynamics \eqref{eq1}) versus $M$ evaders (with dynamics \eqref{eq2}), the communication among the players are described by $\mathcal{G}_p$, $\mathcal{G}_e$ and $\mathcal{G}_{pe}$. 
Let each pursuer adopts the control policy \eqref{eq25}, with the optimal value function $V_i^p$ satisfying the coupled HJI equation \eqref{eq29}, and each evader adopts the control policy \eqref{eq26}, with the optimal value function $V_j^e$ satisfying the coupled HJI equation \eqref{eq30}, then the capture is achieved for each pursuer if its error dynamics \eqref{eq4} is stable.
To this end, the tuple of optimal policies of all players constitute slightly altruistic global Nash equilibrium.
\end{theorem}
\begin{proof}
(i) \textbf{Capture:} Choose the optimal value function $V_i^p$ as the Lyapunov function candidate, then its time derivative is derived by
\begin{equation}{\label{eq31}}
\begin{aligned}
\dot{V}_i^p = & {\nabla{V_i^{p}}}^{\mathrm{T}} \bm{\delta}_i^p \\
                  = & {\nabla{J_i^{p}}}^{\mathrm{T}}   \left[ \bar{A} {\bm{\delta}}^{p}_{i}  + \bar{B}^{p}_{i} \mathbi{u}_{i} -  \sum\limits_{k \in \mathcal{N}_{-i}^p} c^{p}_{ik} E^{n}_{12} \bar{B}_k^p \mathbi{u}_{k}\right.\\ 
                 & \left. \qquad \qquad \qquad \qquad\qquad  -  \sum\limits_{j \in \mathcal{N}^{e}} c^{pe}_{ij} E^{n}_{22} \bar{B}_j^e \mathbi{v}_{j} \right]\\
                  =&  -{\bm{\delta}^p_i}^{\mathrm{T}} \tilde{Q}^{p}_{i} {\bm{\delta}^{p}_i} \!-\! {\bm{\delta}^p_i}^{\mathrm{T}}  \! \left( \sum\limits_{k \in \mathcal{N}_{-i}^p} \! c^{p}_{ik} Q_{ik}^p {\bm{\delta}^p_k}
  \!+\!  \sum\limits_{j \in \mathcal{N}^e} \!c^{pe}_{ij} Q_{ij}^{pe}  {\bm{\delta}^{e}_j} 
\!  \right) \\
 & - U(\mathbi{u}^*_i) - \sum\limits_{k \in \mathcal{N}_{-i}^p} c^{p}_{ik} U(\mathbi{u}^*_k) + \sum\limits_{j \in \mathcal{N}^e} c^{pe}_{ij} U(\mathbi{v}^*_j) - \rho_i.
\end{aligned}
\end{equation}
Based on \eqref{eq31}, $\dot{V}_i^p$ is negative on the condition that
\begin{equation}{\label{eq32}}
 U(\mathbi{u}_i) + \sum\limits_{k \in \mathcal{N}_{-i}^p} c^{p}_{ik} U(\mathbi{u}_k) - \sum\limits_{j \in \mathcal{N}^e} c^{pe}_{ij} U(\mathbi{v}_j) >0,
\end{equation}
which is the capture condition. Hence, the local dynamics of each pursuer is stable if \eqref{eq32} is satisfied, i.e., the pursuers capture their neighboring adversarial players.
Note that the capture condition \eqref{eq32} can be satisfied through target selection strategy, which is further studied in Section V.

(ii) \textbf{Global Nash equilibrium:} Consider the proposed slightly altruistic index \eqref{eq17} for pursuer $i$, it could be rewritten by
\begin{equation}{\label{eq33}}
\begin{aligned}
J_i^p = J_i^p +  \int_{0}^{t_f} \dot{V}_i^p \mathrm{d}t + V_i^p \left(\bm{\delta}_i^p(0) \right) - V_i^p \left(\bm{\delta}_i^p(t_f) \right).
\end{aligned}
\end{equation}
Due to the fact that 
$V_i^p \left(\bm{\delta}_i^p(t_f) \right) = {\bm{\delta}^p_i}^{\mathrm{T}}(t_f) \bar{\tilde{Q}}^{p}_{i} {\bm{\delta}^{p}_i}(t_f) 
+{\bm{\delta}^p_i}^{\mathrm{T}}(t_f)  \left( \sum\nolimits_{k \in \mathcal{N}_{-i}^p} c^{p}_{ik} \bar{Q}_{ik}^p {\bm{\delta}^p_k}(t_f)
  + \sum\nolimits_{j \in \mathcal{N}^e} c^{pe}_{ij} \bar{Q}_{ij}^{pe}  {\bm{\delta}^{e}_j}(t_f) 
 \right)$, then \eqref{eq33} is further converted to
\begin{equation}{\label{eq34}}
\begin{aligned}
J_i^p \!=\! & 
 \int_{0}^{t_f} \!\!
 \left[ {\bm{\delta}^p_i}^{\mathrm{T}} \tilde{Q}^{p}_{i} {\bm{\delta}^{p}_i} + {\bm{\delta}^p_i}^{\mathrm{T}} \! \left( \sum\limits_{k \in \mathcal{N}_{-i}^p} \!\! c^{p}_{ik} Q_{ik}^p {\bm{\delta}^p_k}
  + \!  \sum\limits_{j \in \mathcal{N}^e} \!\! c^{pe}_{ij} Q_{ij}^{pe}  {\bm{\delta}^{e}_j} \!\!
 \right) \right. \\
 &+ 2 \int_{0}^{\mathbi{u}_i} \left( u_{\max} \tanh^{-1} \left( \frac{\bm{\nu}} {u_{\max}} \right) \right)^{\mathrm{T}} R^p_i \ \mathrm{d} \bm{\nu}\\
 & + \sum\limits_{k \in \mathcal{N}_{-i}^p} 2c^{p}_{ik}  \int_{0}^{\mathbi{u}_k} \left( u_{\max} \tanh^{-1} \left( \frac{\bm{\nu}} {u_{\max}} \right) \right)^{\mathrm{T}} R^p_k \ \mathrm{d} \bm{\nu}\\
 & \left. - \sum\limits_{j \in \mathcal{N}^e} 2c^{pe}_{ij}  \int_{0}^{\mathbi{v}_j} \left( v_{\max} \tanh^{-1} \left( \frac{\bm{\nu}} {v_{\max}} \right) \right)^{\mathrm{T}} R^e_j \ \mathrm{d} \bm{\nu} + \rho_i \right]  \mathrm{d}t \\
 & + \int_{0}^{t_f} {\nabla{V_i^{p}}}^{\mathrm{T}} \bm{\delta}_i^p  \mathrm{d}t +  V_i^p \left(\bm{\delta}_i^p(0)\right).
\end{aligned}
\end{equation}

By combining with \eqref{eq31}, \eqref{eq34} is further converted to
\begin{equation}{\label{eq35}}
\begin{aligned}
J_i^p \!=\! & 
 \int_{0}^{t_f} 
 \left[ 
  2 \int_{0}^{\mathbi{u}_i} \left( u_{\max} \tanh^{-1} \left( \frac{\bm{\nu}} {u_{\max}} \right) \right)^{\mathrm{T}} R^p_i \ \mathrm{d} \bm{\nu} \right. \\
 & -2 \int_{0}^{\mathbi{u}^*_i} \left( u_{\max} \tanh^{-1} \left( \frac{\bm{\nu}} {u_{\max}} \right) \right)^{\mathrm{T}} R^p_i \ \mathrm{d} \bm{\nu} \\
 & + \sum\limits_{k \in \mathcal{N}_{-i}^p} 2c^{p}_{ik}  \int_{0}^{\mathbi{u}_k} \left( u_{\max} \tanh^{-1} \left( \frac{\bm{\nu}} {u_{\max}} \right) \right)^{\mathrm{T}} R^p_k \ \mathrm{d} \bm{\nu}\\
 & - \sum\limits_{k \in \mathcal{N}_{-i}^p} 2c^{p}_{ik}  \int_{0}^{\mathbi{u}^*_k} \left( u_{\max} \tanh^{-1} \left( \frac{\bm{\nu}} {u_{\max}} \right) \right)^{\mathrm{T}} R^p_k \ \mathrm{d} \bm{\nu}\\
 & - \sum\limits_{j \in \mathcal{N}^e} 2c^{pe}_{ij}  \int_{0}^{\mathbi{v}_j} \left( v_{\max} \tanh^{-1} \left( \frac{\bm{\nu}} {v_{\max}} \right) \right)^{\mathrm{T}} R^e_j \ \mathrm{d} \bm{\nu}\\
 & \left. + \sum\limits_{j \in \mathcal{N}^e} 2c^{pe}_{ij}  \int_{0}^{\mathbi{v}^*_j} \left( v_{\max} \tanh^{-1} \left( \frac{\bm{\nu}} {v_{\max}} \right) \right)^{\mathrm{T}} R^e_j \ \mathrm{d} \bm{\nu}  \right]  \mathrm{d}t \\
 & +  V_i^p \left(\bm{\delta}_i^p(0)\right).
\end{aligned}
\end{equation}
Consider the fact that $\int_{0}^{\mathbi{u}_i} \left( u_{\max} \tanh^{-1} \left( {\bm{\nu}}/ {u_{\max}} \right) \right)^{\mathrm{T}} R^p_i \ \mathrm{d} \bm{\nu}
- \int_{0}^{\mathbi{u}^*_i} \left( u_{\max} \tanh^{-1} \left( {\bm{\nu}} /{u_{\max}} \right) \right)^{\mathrm{T}} R^p_i \ \mathrm{d} \bm{\nu} = \int_{\mathbi{u}*_i}^{\mathbi{u}_i} \left( u_{\max} \tanh^{-1} \left( {\bm{\nu}}/ {u_{\max}} \right) \right)^{\mathrm{T}} R^p_i \ \mathrm{d} \bm{\nu}$, and the other terms have the similar form. According to the result in[], it could be conclude that the above equation is semi-positive, it equals to zeros if and only if ${\mathbi{u}*_i}={\mathbi{u}_i}$, then it could be concluded that the value function $J_i^p$ has the optimal value of $V_i^p \left(\bm{\delta}_i^p(0)\right)$, on the condition that $\mathbi{u}_i = \mathbi{u}^*_i$ and the neighboring players adopt their optimal policy.

Similarly, the value function for evader $j$ could be converted to the following form, i.e.,
\begin{equation}{\label{eq36}}
\begin{aligned}
J_j^e \!=\! & 
 \int_{0}^{t_f} 
 \left[ 
  2 \int_{\mathbi{v}^*_j}^{\mathbi{v}_j} \left( v_{\max} \tanh^{-1} \left( \frac{\bm{\nu}} {v_{\max}} \right) \right)^{\mathrm{T}} R^e_j \ \mathrm{d} \bm{\nu} \right. \\
 & + \sum\limits_{l \in \mathcal{N}_{-j}^e} c^{e}_{jl} 2 \int_{\mathbi{v}^*_l}^{\mathbi{v}_l} \left( v_{\max} \tanh^{-1} \left( \frac{\bm{\nu}} {v_{\max}} \right) \right)^{\mathrm{T}} R^e_l \ \mathrm{d} \bm{\nu}\\
 & \left. - \sum\limits_{i \in \mathcal{N}^p} c^{ep}_{ji} 2 \int_{\mathbi{u}^*_i}^{\mathbi{u}_i} \left( u_{\max} \tanh^{-1} \left( \frac{\bm{\nu}} {u_{\max}} \right) \right)^{\mathrm{T}} R^p_i \ \mathrm{d} \bm{\nu}  \right]  \mathrm{d}t \\
 & +  V_j^e \left(\bm{\delta}_j^e(0)\right).
\end{aligned}
\end{equation}
Based on \eqref{eq35} and \eqref{eq36}, it could be concluded that the global Nash equilibrium
\begin{equation*}
\begin{aligned}
 J_{i}^{p} ({ \mathbi{u}^{*}_i }, {\left. {\mathbi{u}}^{*}_{\mathcal{G}_p-i} \right.}, {\left. {\mathbi{v}}^{*}_{\mathcal{G}_{e}} \right.}) &\leq J_{i}^{p} (  { \mathbi{u}_i }, {\left. {\mathbi{u}}^{*}_{\mathcal{G}_p-i} \right.}, {\left. {\mathbi{v}}^{*}_{\mathcal{G}_{e}} \right.})\\
J_{j}^{e} ({ \mathbi{v}^{*}_j }, {\left. {\mathbi{v}}^{*}_{\mathcal{G}_e-j} \right.}, {\left. {\mathbi{u}}^{*}_{\mathcal{G}_{p}} \right.}) &\leq J_{i}^{p} (  { \mathbi{v}_i }, {\left. {\mathbi{v}}^{*}_{\mathcal{G}_e-j} \right.}, {\left. {\mathbi{u}}^{*}_{\mathcal{G}_{p}} \right.})
\end{aligned}
\end{equation*}
is achieved. $\hfill\blacksquare$
\end{proof} 

\section{Solving Constrained MPE Games via Reinforcement Learning}
\subsection{Policy iteration algorithm}
Consider the coupled HJI equations \eqref{eq29} and \eqref{eq30} with input constraints can not solved by the Riccati equation-based methods.
The actor-critic reinforcement learning is adopted in this paper to solve the coupled HJI equations.
A policy iteration algorithm shown in Algorithm \ref{algorithm1} applied in the MPE games with slight altruism is first provided.
\begin{algorithm}[t]
	\caption{Policy Iteration for Input-constrained MPE Games.}
	\label{algorithm1}
	\KwIn{ The real-time local error of each player and its neighboring players.}
	\KwOut{The tuple of optimal control policies of all players.}  
	\BlankLine
	Initializing control policy $\mathbi{u}_i^k$ and $\mathbi{v}_j^k, k=0$ satisfying stabilization, $\forall i \in \mathcal{N}^p$, $\forall j \in \mathcal{N}^e$;
	
	\While{\textnormal{$ \| \left({V_i^p}\right)^{k+1} - \left({V_i^p}\right)^{k} \| >  \gamma$}}{

		\ForEach{k}{
			\emph{Step 1 (Policy  Evaluation):} solve the value functions $\left({V_i^p}\right)^{k}$ and $\left({V_j^e}\right)^{k}$ using
			\begin{equation}{\label{eq37}}
			\begin{aligned}
			H_i^p \left( \left({V_i^p}\right)^{k}, \bm{\delta}_i^p, \bm{\delta}_{\mathcal{G}_p-i}^p, \bm{\delta}_{\mathcal{G}_{pe}-i}^e, \right. \\ \left.  { \mathbi{u}^k_i },  {\mathbi{u}}^{k}_{\mathcal{G}_p-i}, {\mathbi{v}}^{k}_{\mathcal{G}_{pe}-i}  \right) =0,
			\end{aligned}
			\end{equation}
			and
			\begin{equation}{\label{eq38}}
			\begin{aligned}
			H_j^e \left( \left({V_j^e}\right)^{k}, \bm{\delta}_j^e, \bm{\delta}_{\mathcal{G}_e-j}^e, \bm{\delta}_{\mathcal{G}_{pe}-j}^p, \right. \\ \left.  { \mathbi{v}^k_j },  {\mathbi{v}}^{k}_{\mathcal{G}_e-j}, {\mathbi{u}}^{k}_{\mathcal{G}_{pe}-j}  \right) =0;
			\end{aligned}
			\end{equation}
			
			\emph{Step 2 (Policy Improvement):} The control policies are updated by
			\begin{equation}{\label{eq39}}
			\begin{aligned}
			\mathbi{u}_i^{k+1} &= \arg \min H_i^p \\
			& \Downarrow\\
			 {\bm{u}_i^{k+1}} &\!=\! -u_{\max} \! \tanh \! \left(\! \frac{1}{2u_{\max}} {R_i^p}^{-1} {\left.\bar{B}_i^p\right.}^{\mathrm{T}} {\left( \nabla{V_i^{p}}\right)}^k \! \right)
			\end{aligned}
			\end{equation}
			and
			\begin{equation}{\label{eq40}}
			\begin{aligned}
			\mathbi{v}_j^{k+1} &= \arg \min H_j^e \\
			& \Downarrow\\
			 {\bm{v}_j^{k+1}} &\!= \!-v_{\max} \! \tanh \! \left(\! \frac{1}{2v_{\max}} {R_j^e}^{-1} {\left.\bar{B}_j^e\right.}^{\mathrm{T}} {\left( \nabla{V_j^{e}}  \right)}^k \! \right)
			\end{aligned}
			\end{equation}
		}
		Update $k$ by $k+1$.
	}

\end{algorithm}

\begin{theorem}[Convergence of Policy Iteration Algorithm]
\label{theorem2}
Suppose each player in the input-constrained MPE games performs Algorithm \ref{algorithm1}, then the iterated value functions converge to their corresponding optimal values $V_i^p\ (i \in \mathcal{N}_p)$ and $V_j^e\ (j \in \mathcal{N}_e)$, and the iterated control policies converge to the optimal policies, under the both following cases, i.e., \\
(i) Only the player performing the algorithm updates his control policy;\\
(ii) All the players update their control policy.
\end{theorem}
\begin{proof}
By integrating the proof of input-constrained optimal control for a single agent in \cite{ref2,ref3}, the remaining proof of Theorem \ref{theorem2} is similar with \cite{ref1}. $\hfill\blacksquare$
\end{proof}

\subsection{Critic-actor neural network}
Based on the policy iteration Algorithm \ref{algorithm1}, the following results provide the online adaptive learned solution for the MPE game.
The learning process is based on the architecture of a critic neural network (NN) which is used for approximating the optimal value functions, and an actor NN which is employed to approximate the optimal control policies.
The tuning law provided later is used to adjust the weight of the neural networks, which are introduced first.

\textbf{Critic NN:} Based on the Weierstrass higher-order approximation theorem, the value function of the $i$-th pursuer could be represented by
\begin{equation}{\label{eq41}}
V_i^p  = \left({W_i^p}\right)^{\mathrm{T}} \phi_i^p (\bm{\delta}_i^p, \bm{\delta}_{\mathcal{G}_p-i}^p, \bm{\delta}_{\mathcal{G}_{pe}-i}^e) + \varepsilon_i^p
\end{equation}
where ${W_i^p} \in  \mathbb{R}^{h}  $ denotes the weight vector, $\phi_i^p (\bm{\delta}_i^p, \bm{\delta}_{\mathcal{G}_p-i}^p, \bm{\delta}_{\mathcal{G}_{pe}-i}^e)$ represents the basis function vector of the neural network with $h$ hidden layers, $\varepsilon_i^p$ represents the approximating error. Note that the basis function vector is related to both the local of pursuer $i$, but also the local errors of his neighboring players, due to the proposed slightly altruistic index function. The objective of critic NN is using the estimation of $W_i^p$ to approximate its corresponding actual value, such that
\begin{equation}{\label{eq42}}
\hat{V}_i^p ( {\bm{\delta}^{p}_i} ) = \left({\hat{W}_{i|c}^p}\right)^{\mathrm{T}} \phi_i^p (\bm{\delta}_i^p, \bm{\delta}_{\mathcal{G}_p-i}^p, \bm{\delta}_{\mathcal{G}_{pe}-i}^e) \Rightarrow W_{i|c}^p \phi_i^p  .
\end{equation}

Based on \eqref{eq21} and \eqref{eq42}, the Bellman equation is converted to
\begin{equation}{\label{eq43}}
\begin{aligned}
H_i^p = &   {\bm{\delta}^p_i}^{\mathrm{T}} \tilde{Q}^{p}_{i} {\bm{\delta}^{p}_i} + {\bm{\delta}^p_i}^{\mathrm{T}}  \left( \sum\limits_{k \in \mathcal{N}_{-i}^p} c^{p}_{ik} Q_{ik}^p {\bm{\delta}^p_k}
  - \sum\limits_{j \in \mathcal{N}^e} c^{pe}_{ij} Q_{ij}^{pe}  {\bm{\delta}^{e}_j} 
 \right) \\
 &+ U(\mathbi{u}_i) + \sum\limits_{k \in \mathcal{N}_{-i}^p} c^{p}_{ik} U(\mathbi{u}_k) - \sum\limits_{j \in \mathcal{N}^e} c^{pe}_{ij} U(\mathbi{v}_j) + \rho \\
                & +\left({\hat{W}_{i|c}^p}\right)^{\mathrm{T}} \phi_i^p   \left[ \bar{A} {\bm{\delta}}^{p}_{i}  + \bar{B}^{p}_{i} \mathbi{u}_{i} -  \sum\limits_{k \in \mathcal{N}_{-i}^p} c^{p}_{ik} E^{n}_{12} B_k^p \mathbi{u}_{k}\right. \\
                &\left. -  \sum\limits_{j \in \mathcal{N}^{e}} c^{pe}_{ij} E^{n}_{22} B_j^e \mathbi{v}_{j} \right] \equiv \zeta_{i|c}^p. 
\end{aligned}
\end{equation}

Define the optimization objective of the critic NN as
\begin{equation}{\label{eq44}}
E_{i|c}^p = \frac{1}{2} \left( \zeta_{i|c}^p \right)^2,
\end{equation}
where the least square solution of \eqref{eq44} is obtained by $\hat{W}_{i|c}^p \rightarrow W_i^p$. The tuning law of the critic weights is provided in Theorem \ref{theorem3}. 

Similarly, the critic NN for evader $j$ is described by
\begin{equation}{\label{eq45}}
V_j^e  = \left({W_j^e}\right)^{\mathrm{T}} \phi_j^e (\bm{\delta}_j^e, \bm{\delta}_{\mathcal{G}_e-j}^e, \bm{\delta}_{\mathcal{G}_{pe}-j}^p) + \varepsilon_j^e,
\end{equation}
where the parameters are defined similarly with \eqref{eq41}. By defining the approximating error as $\zeta_j^e$ which is similar with \eqref{eq43}, the optimization function is described by
\begin{equation}{\label{eq46}}
E_{j|c}^e = \frac{1}{2} \left( \zeta_{j|c}^e \right)^2,
\end{equation}
where the least square solution of \eqref{eq46} is obtained by $\hat{W}_{j|c}^e \rightarrow W_j^e$. The tuning law of the critic weights is also provided in Theorem \ref{theorem3}.

\textbf{Actor NN:} Based on the actor neural network, the control law of pursuer $i$ that approximates \eqref{eq25} is described by
\begin{equation}{\label{eq47}}
{\hat{\bm{u}}_i^{*}} = -u_{\max}\tanh \left(\frac{1}{2u_{\max}} {R_i^p}^{-1} {B_i^p}^{\mathrm{T}}  {\frac{\partial{\phi_i^p}}{\partial{\bm{\delta}^p_i}}}^{\mathrm{T}}  \hat{W}_{i|a}^p \right).
\end{equation}

Define the error of the actor NN as
\begin{equation}{\label{eq48}}
\begin{aligned}
\bm{\zeta}_{i|a}^p = & {\hat{\bm{u}}_i^{*}} - {\bm{u}_i^{*}}\\
                             = & u_{\max} \left[ \tanh\left(\frac{1}{2u_{\max}} {R_i^p}^{-1} {B_i^p}^{\mathrm{T}}  {\frac{\partial{\phi_i^p}}{\partial{\bm{\delta}^p_i}}}^{\mathrm{T}}  \hat{W}_{i|a}^p \right) \right.\\
                             & \left. - \tanh\left(\frac{1}{2u_{\max}} {R_i^p}^{-1} {B_i^p}^{\mathrm{T}}  {\frac{\partial{\phi_i^p}}{\partial{\bm{\delta}^p_i}}}^{\mathrm{T}}  \hat{W}_{i|c}^p \right)
                             \right].
\end{aligned}
\end{equation}
Thus define the optimization objective of the actor NN as
\begin{equation}{\label{eq49}}
E_{i|c}^p = {\bm{\zeta}_{i|a}^p}^{\mathrm{T}} R_i^p \bm{\zeta}_{i|a}^p,
\end{equation}
with the tuning law of the actor weights is provided in Theorem \ref{theorem3}. 

Similarly, the control law of evader $j$ that approximates \eqref{eq26} is described by
\begin{equation}{\label{eq50}}
{\hat{\bm{v}}_j^{*}} = -v_{\max}\tanh \left(\frac{1}{2v_{\max}} {R_j^e}^{-1} {B_j^e}^{\mathrm{T}}  {\frac{\partial{\phi_j^e}}{\partial{\bm{\delta}^e_j}}}^{\mathrm{T}}  \hat{W}_{j|a}^e \right).
\end{equation}
Define the error of the actor NN as
\begin{equation}{\label{eq51}}
\begin{aligned}
\bm{\zeta}_{j|a}^e = & {\hat{\bm{v}}_j^{*}} - {\bm{v}_j^{*}}\\
                             = & v_{\max} \left[ \tanh\left(\frac{1}{2v_{\max}} {R_j^e}^{-1} {B_j^e}^{\mathrm{T}}  {\frac{\partial{\phi_j^e}}{\partial{\bm{\delta}^e_j}}}^{\mathrm{T}}  \hat{W}_{j|a}^e \right) \right.\\
                             & \left. - \tanh\left(\frac{1}{2v_{\max}} {R_j^e}^{-1} {B_j^e}^{\mathrm{T}}  {\frac{\partial{\phi_j^e}}{\partial{\bm{\delta}^e_j}}}^{\mathrm{T}}  \hat{W}_{j|c}^e \right)
                             \right].
\end{aligned}
\end{equation}
Thus define the optimization objective of the actor NN as
\begin{equation}{\label{eq52}}
E_{j|c}^e = {\bm{\zeta}_{j|a}^e}^{\mathrm{T}} R_j^e \bm{\zeta}_{j|a}^e,
\end{equation}
with the tuning law of the actor weights is also provided in Theorem \ref{theorem3}.

\subsection{Tuning law}
\textbf{Critic NN:} Based on \eqref{eq46}, the tuning law for the critic NN of pursuer $i$ is described by
\begin{equation}{\label{eq53}}
\begin{aligned}
\dot{\hat{W}}_{i|c}^p \!=\! & -\alpha_i \frac{\partial{E_{i|c}^p}}{\partial{{\hat{W}_{i|c}^p}}}\\
                            = & -\alpha_i \frac{{\sigma_{i|c}^p}}{\left(1+ {{\sigma_{i|c}^p}}^{\mathrm{T}} {\sigma_{i|c}^p} \right)^2} 
                            \left[ \left( {\sigma_{i|c}^p} \right)^{\mathrm{T}} {\hat{W}_{i|c}^p}  \right.\\
                            & +  {\bm{\delta}^p_i}^{\mathrm{T}} \tilde{Q}^{p}_{i} {\bm{\delta}^{p}_i} + {\bm{\delta}^p_i}^{\mathrm{T}}  \left( \sum\limits_{k \in \mathcal{N}_{-i}^p} c^{p}_{ik} Q_{ik}^p {\bm{\delta}^p_k}
  - \sum\limits_{j \in \mathcal{N}^e} c^{pe}_{ij} Q_{ij}^{pe}  {\bm{\delta}^{e}_j}  \right)\\
  &\left. + U(\hat{\mathbi{u}}_i) + \sum\limits_{k \in \mathcal{N}_{-i}^p} c^{p}_{ik} U(\hat{\mathbi{u}}_k) - \sum\limits_{j \in \mathcal{N}^e} c^{pe}_{ij} U(\hat{\mathbi{v}}_j) + \rho_i \right],
\end{aligned}
\end{equation}
where $\alpha_i$ denotes a positive tuning parameter, ${\sigma_{i|c}^p} = \frac{\partial{\phi_i^p}}{\partial{\bm{\delta}^p_i}} \left( \bar{A} {\bm{\delta}}^{p}_{i}  + \bar{B}^{p}_{i} \hat{\mathbi{u}}_{i|a} -  \sum\limits_{k \in \mathcal{N}_{-i}^p} c^{p}_{ik} E^{n}_{12} B_k^p \hat{\mathbi{u}}_{k|a} \right.$ $\left. -  \sum\limits_{j \in \mathcal{N}^{e}} c^{pe}_{ij} E^{n}_{22} B_j^e \hat{\mathbi{v}}_{j|a} \right)$, and
\begin{equation}{\label{eq54}}
\begin{aligned}
U(\hat{\mathbi{u}}_i) = & 2 \int_{0}^{\hat{\mathbi{u}}_i} \left( u_{\max} \tanh^{-1} (\frac{\bm{\nu}}{u_{\max}}) \right)^{\mathrm{T}} R^p_i \ \mathrm{d} \bm{\nu} \\
                                 = & u_{\max} \left( {\hat{W}_{i|a}^p} \right)^{\mathrm{T}} \frac{\partial{\phi_i^p}}{\partial{\bm{\delta}^p_i}} B_i^p \tanh(\hat{D}_i^p)\\
                                 &+ u^2_{\max} \bar{R}_i^p \ln \left( \underbar{1} - \tanh^2(\hat{D}_i^p)\right),
\end{aligned}
\end{equation}
with $\hat{D}_i^p = \frac{1}{2 u_{\max}} {R_i^p}^{-1} {B_i^p}^{\mathrm{T}} {\frac{\partial{\phi_i^p}}{\partial{\bm{\delta}^p_i}}}^{\mathrm{T}} \hat{W}_{i|a}^p$.

Similarly, the tuning law for the critic NN of evader $j$ is described by
\begin{equation}{\label{eq55}}
\begin{aligned}
\dot{\hat{W}}_{j|c}^e \!=\! & -\alpha_j \frac{\partial{E_{j|c}^e}}{\partial{{\hat{W}_{j|c}^e}}}\\
                            = & -\alpha_j \frac{{\sigma_{j|c}^e}}{\left(1+ {{\sigma_{j|c}^e}}^{\mathrm{T}} {\sigma_{j|c}^e} \right)^2} 
                            \left[ \left( {\sigma_{j|c}^e} \right)^{\mathrm{T}} {\hat{W}_{j|c}^e}  \right.\\
                            & -  {\bm{\delta}^e_j}^{\mathrm{T}} \tilde{Q}^{e}_{j} {\bm{\delta}^{e}_j} - {\bm{\delta}^e_j}^{\mathrm{T}}  \left( \sum\limits_{l \in \mathcal{N}_{-j}^e} c^{e}_{jl} Q_{jl}^e {\bm{\delta}^e_l}
  - \sum\limits_{i \in \mathcal{N}^p} c^{ep}_{ji} Q_{ji}^{ep}  {\bm{\delta}^{p}_i}  \right)\\
  &\left. + U(\hat{\mathbi{v}}_i) + \sum\limits_{l \in \mathcal{N}_{-j}^e} c^{e}_{jl} U(\hat{\mathbi{v}}_l) - \sum\limits_{i \in \mathcal{N}^p} c^{ep}_{ji} U(\hat{\mathbi{v}}_i) + \rho_j \right],
\end{aligned}
\end{equation}
where ${\sigma_{j|c}^e} = \frac{\partial{\phi_j^e}}{\partial{\bm{\delta}^e_j}} \left( \bar{A} {\bm{\delta}}^{e}_{j}  + \bar{B}^{e}_{j} \hat{\mathbi{v}}_{j|a} -  \sum\limits_{l \in \mathcal{N}_{-j}^e} c^{e}_{jl} E^{n}_{12} B_l^e \hat{\mathbi{v}}_{l|a} \right.$ $\left. -  \sum\limits_{i \in \mathcal{N}^{p}} c^{ep}_{ji} E^{n}_{22} B_i^p \hat{\mathbi{u}}_{i|a} \right)$, and
\begin{equation}{\label{eq56}}
\begin{aligned}
U(\hat{\mathbi{v}}_j) = & 2 \int_{0}^{\hat{\mathbi{v}}_j} \left( v_{\max} \tanh^{-1} (\frac{\bm{\nu}}{v_{\max}}) \right)^{\mathrm{T}} R_j^e \ \mathrm{d} \bm{\nu} \\
                                 = & v_{\max} \left( {\hat{W}_{j|a}^e} \right)^{\mathrm{T}} \frac{\partial{\phi_j^e}}{\partial{\bm{\delta}^e_j}} B_i^p \tanh(\hat{D}_j^e)\\
                                 &+ v^2_{\max} \bar{R}_j^e \ln \left( \underbar{1} - \tanh^2(\hat{D}_j^e)\right)
\end{aligned}
\end{equation}
with $\hat{D}_j^e = \frac{1}{2 v_{\max}} {R_j^e}^{-1} {B_j^e}^{\mathrm{T}} {\frac{\partial{\phi_j^e}}{\partial{\bm{\delta}^e_j}}}^{\mathrm{T}} \hat{W}_{j|a}^e$.

\textbf{Actor NN:} The tuning law of the actor NN for pursuer $i$ is 
\begin{equation}{\label{eq57}}
\begin{aligned}
\dot{\hat{W}}_{i|a}^p \!=\! & - \beta_i \left( {\frac{\partial{\phi_i^p}}{\partial{\bm{\delta}^p_i}}} B_i^p \bm{\zeta}_{i|a}^p \!+\!  {\frac{\partial{\phi_i^p}}{\partial{\bm{\delta}^p_i}}} B_i^p \tanh^2 (\hat{D}_i^p) \bm{\zeta}_{i|a}^p \!+\! Y_i {\hat{W}}_{i|a}^p
\right)
\end{aligned}
\end{equation}
where $\beta_i $ denotes a positive tuning parameter, $Y_i$ is a design parameter to assure stabilit.

The tuning law of the actor NN for evader $j$ is 
\begin{equation}{\label{eq58}}
\begin{aligned}
\dot{\hat{W}}_{j|a}^e \!=\! & - \beta_j \left( {\frac{\partial{\phi_j^e}}{\partial{\bm{\delta}^e_j}}} B_j^e \bm{\zeta}_{j|a}^e \!+\!  {\frac{\partial{\phi_j^e}}{\partial{\bm{\delta}^e_j}}} B_j^e \tanh^2 (\hat{D}_j^e) \bm{\zeta}_{j|a}^e \!+\! Y_j {\hat{W}}_{j|a}^e
\right)
\end{aligned}
\end{equation}
where $\beta_j $ denotes a positive tuning parameter, $Y_j$ is a design parameter to assure stability.

\begin{theorem}
\label{theorem3}
Consider the input-constrained MPE games with the pursuers' local error dynamics described \eqref{eq3} and the evaders' local error dynamics given by \eqref{eq5}.
Let the critic NN of each pursuer is described by \eqref{eq42}, and the the control input generated by actor NN is given by \eqref{eq47}, if the tuning laws for critic NN and actor NN are given by 
\eqref{eq53} and \eqref{eq57}, then the approximating error of both the critic NN and actor NN are uniformly ultimately bounded. Similarly, let each evader adopts the critic NN and actor NN given by \eqref{eq45} and \eqref{eq50}, with the tuning law given by \eqref{eq55} and \eqref{eq58}, then the approximating errors are also uniformly ultimately bounded.
\end{theorem}
\begin{proof}
The proof is similar with \cite{ref1}.$\hfill\blacksquare$
\end{proof}

\begin{theorem}[Convergence to Zero-Sum Nash Equilibrium]
Suppose that Theorem \ref{theorem3} holds, then 

(i) $H_i^p \left( \left({\hat{W}_{i|c}^p}\right)^{\mathrm{T}} \phi_i^p, \bm{\delta}_i^p, \bm{\delta}_{\mathcal{G}_p-i}^p, \bm{\delta}_{\mathcal{G}_{pe}-i}^e,  \hat{ \mathbi{u}_i },  \hat{\mathbi{u}}_{\mathcal{G}_p-i}, \hat{\mathbi{v}}_{\mathcal{G}_{pe}-i}  \right)$ converges to the approximate solution for HJI equation \eqref{eq29}
$H_j^e \left( \left({\hat{W}_{j|c}^e}\right)^{\mathrm{T}} \phi_j^e, \bm{\delta}_j^e, \bm{\delta}_{\mathcal{G}_e-j}^e, \bm{\delta}_{\mathcal{G}_{pe}-j}^p,  { \hat{\mathbi{v}}_j },  \hat{{\mathbi{v}}}_{\mathcal{G}_e-j}, \hat{\mathbi{u}}_{\mathcal{G}_{pe}-j}  \right) $ converges to the approximate solution for HJI equation \eqref{eq30}.

(ii) All the control policies converge to the approximate global Nash equilibrium.
\end{theorem}
\begin{proof}
The proof is similar with \cite{ref1}.$\hfill\blacksquare$
\end{proof}

\section{Rolling Horizon Target Selection and Capture Analysis}
\label{sec5}
Consider the capture conditions, which lead to the local dynamics of pursuers are stable.
In \eqref{eq32}, the condition for pursuer $i$ is given by letting the summation of control energy of pursuer $i$ and his neighboring cooperative teammates is larger than his neighboring opponents.
It should be noted that such condition in existing literature is satisfied by the selection of the positive-definite control weight matrix $R$ in a quadratic energy form such that $\mathbi{u}^{\mathrm{T}} R \mathbi{u}$.
Through letting the control weight matrix of the evader is larger than his pursuer, then the capture could occurs, which means the pursuer puts less attention on energy saving.
However, the input constraints are not considered.
Due to the integral form of control energy in \eqref{eq11}, the capture conditions could not be simply described by the selection of control weight matrix.

Besides, the capture conditions restrict the maneuver ability of pursuers is bigger than the evaders, ignoring that multiple players with poor mobility can capture fewer players with strong mobility.
Such assumption is overcame is this paper, and the capture condition of multiple player with poor mobility pursuit fewer players with strong mobility is studied.
\subsection{Bi-layer topology and rolling horizon target selection}
In this paper, a bi-layer topology is proposed, in which the first layer contains the communication information, and the second layer called game topology guarantees the capture can be achieved.
Besides, the time element in the finite-time games is novel considered in the game topology, i.e., as time converges the setting terminal, the pursuers discard the evaders that are difficult to pursuit, and select the evaders that the capture is easier to be achieved.

Based on the above description and Section \ref{Graph Theory}, the topology are adjusted by replacing each edge in $\mathcal{G}_{pe}$ to $\tilde{c}^{pe}_{ij} = c_{ij}^{pe} \times g_{ij}^{pe}$, where $g^p_{ik} = 0,1$ is the game weight to guarantee the capture is satisfied. Note that the weights in $\mathcal{G}^{ep}$ is not changed.
The following Algorithm \ref{algorithm2} shows the setting law of the weight in the bi-layer topology.

\begin{algorithm}[t]
	\caption{Adjusting law of weights in bi-layer topology.}
	\label{algorithm2}
	For pursuer $i$, initialize the communication weights, initialize all the weights in the game topology as $1$, initialize the evader set $\Theta_i^{kT}$ involving all his neighboring evaders;
	
	\While{$t \leq t_f$}{
				
		\ForEach{$t \in [kT, (k+1)T]$, $T$ is given time interval}{
			\emph{Step 1 (Reachable domain):} Based on the dynamics \eqref{eq1} for each pursuer, using the maximum control input $u_{\max}$, calculate the reachable domain $\Sigma_i^{kT}$ during the interval $[kT, t_f]$;
			
			\emph{Step 2 (Update the evader set $\Theta_i^{kT}$):} Exclude the neighboring evaders outside of reachable domain $\Sigma_i^{kT}$, and set the game weights with all evaders outside of  $\Theta_i^{kT}$ as $0$;
			
			\emph{Step 3 (Capture guarantee):} According to the relative distance with pursuer $i$, \ForEach{evader $j$ in $\Theta_i^{kT}$ (evader has shorter relative distance is preferential)}{calculate
			\begin{equation}\label{eq59}
			U(\mathbi{u}_i) + \sum\limits_{k \in \mathcal{N}_{-i}^p} c^{p}_{ik} U(\mathbi{u}_k) - \sum\limits_{j \in \Theta_i^{kT}} \tilde{c}^{pe}_{ij} U(\mathbi{v}_j)
			\end{equation}
			until
			\begin{equation}\label{eq60}
			\begin{aligned}
			U(\mathbi{u}_i) + \sum\limits_{k \in \mathcal{N}_{-i}^p} c^{p}_{ik} U(\mathbi{u}_k) - \sum\limits_{j \in \Theta_i^{kT}} \tilde{c}^{pe}_{ij} U(\mathbi{v}_j) \\
			\leq \chi, \ \chi>0
			\end{aligned}
			\end{equation}
			}
			set the game weight with each evader (not calculated in \eqref{eq60}) as $0$;
			
					}
		Break this procedure once pursuer $i$ only has one neighboring evader.
	}

\end{algorithm}

According to Algorithm \ref{algorithm2}, given a time interval $T$, during each $[kT, (k+1)T]$ the weights in the game topology is set as constant, the rolling horizon scheme is conducted during each interval.
The target selection is based on the idea of reachable domain, i.e., according to the agent's dynamics and the upper bound of input saturation, the reachable domain during a finite time interval can be calculated.
Through judging whether the neighboring evaders is in the reachable domain, the weights of edges relative to the evaders outside the domain in the game topology are set as zero, which means the pursuer discards the evaders that can not achieve capture.
Note that, with the time converge to the terminal, the reachable domain is contractive, which means the pursuer put more efforts to the evaders that are easier to capture.
Besides, another function of the game topology is to guarantee the capture is achieved in real-time.
Based on the capture condition \eqref{eq32}, the basic idea is to select the nearest evaders through iteration to keep \eqref{eq60} is satisfied where $\chi$ is a positive constant to guarantee finite-time capture.
Note that once each player only has one neighboring evader with the iteration of target selection algorithm, the procedure is automatically complete.

\subsection{Finite-time capture analysis}
\begin{theorem}[Finite-time capture]
Consider the input-constrained MPE games with the pursuers' local error dynamics described \eqref{eq3} and the evaders' local error dynamics given by \eqref{eq5}.
The control policies of pursuer $i$ and evader $j$ are given by \eqref{eq25} and \eqref{eq26}, if the each pursuer adopts the target selection method presented in Algorithm \ref{algorithm2}, then the finite-time capture is guaranteed.
\end{theorem}
\begin{proof}
For pursuer $i$, the capture condition is presented in \eqref{eq32}, i.e.,
\begin{equation*}
U(\mathbi{u}_i) + \sum\limits_{k \in \mathcal{N}_{-i}^p} c^{p}_{ik} U(\mathbi{u}_k) - \sum\limits_{j \in \mathcal{N}^e} c^{pe}_{ij} U(\mathbi{v}_j) >0.
\end{equation*}
According to \eqref{eq60}, if pursuer $i$ adopts the rolling horizon target selection method, then the time derivative of the optimal value function is described by
 \begin{equation}
\begin{aligned}
\dot{V}_i^p 
                  =&  -{\bm{\delta}^p_i}^{\mathrm{T}} \tilde{Q}^{p}_{i} {\bm{\delta}^{p}_i} \!-\! {\bm{\delta}^p_i}^{\mathrm{T}}  \! \left( \sum\limits_{k \in \mathcal{N}_{-i}^p} \! c^{p}_{ik} Q_{ik}^p {\bm{\delta}^p_k}
  \!+\!  \sum\limits_{j \in \mathcal{N}^e} \!c^{pe}_{ij} Q_{ij}^{pe}  {\bm{\delta}^{e}_j}
\!  \right) \\
 & - U(\mathbi{u}^*_i) - \sum\limits_{k \in \mathcal{N}_{-i}^p} c^{p}_{ik} U(\mathbi{u}^*_k) + \sum\limits_{j \in \mathcal{N}^e} c^{pe}_{ij} U(\mathbi{v}^*_j) - \rho_i\\
 \leq &   -{\bm{\delta}^p_i}^{\mathrm{T}} \tilde{Q}^{p}_{i} {\bm{\delta}^{p}_i} \!-\! {\bm{\delta}^p_i}^{\mathrm{T}}  \! \left( \sum\limits_{k \in \mathcal{N}_{-i}^p} \! c^{p}_{ik} Q_{ik}^p {\bm{\delta}^p_k}
  \!+\!  \sum\limits_{j \in \mathcal{N}^e} \!c^{pe}_{ij} Q_{ij}^{pe}  {\bm{\delta}^{e}_j}
\!  \right) \\
& - \chi -\rho_i
\end{aligned}
\end{equation}
Thus it can be concluded that $\dot{V}_i^p$ is less than a negative parameter, and the capture time could be adjusted by the parameter $\chi$.$\hfill\blacksquare$
\end{proof}

\section{Conclusion}
In this paper, the input-constrained MPE games with time-energy optimality is formulated.
By introducing the altruism terms, a slightly altruistic global Nash equilibrium is proposed to constitute a cooperative-noncooperative MPE games framework.
To guarantee the capture conditions for each pursuer, the rolling horizon target selection approach with a bi-layer topology is proposed.
The proposed methods avoids the assumption that the the evader must have the poor mobility.
Detailed theoretical analysis prove the effectiveness of the proposed methods.

\end{document}